\documentclass{elsarticle}
\usepackage{mathtools}
\usepackage{bbm}
\usepackage{float}
\usepackage{color}
\usepackage{comment}
\usepackage{graphicx}
\usepackage{enumerate}
\usepackage[normalem]{ulem}
\usepackage[title,titletoc,toc]{appendix}
\usepackage{titlesec}
\usepackage{amsmath}
\usepackage{amssymb}
\usepackage{amsthm}

\titleformat{\paragraph}
{\normalfont\normalsize\bfseries}{\theparagraph}{1em}{}
\titlespacing*{\paragraph}
{0pt}{3.25ex plus 1ex minus .2ex}{1.5ex plus .2ex}

\DeclareMathOperator*{\argmax}{arg\,max}

\newtheorem{proposition}{Proposition}

\newdefinition{definition}{Definition}
\newdefinition{assumption}{Assumption}
\newdefinition{remark}{Remark}
\newdefinition{example}{Example}

\newcommand{\vb}{\vspace{3mm}}

\begin{document}

\begin{frontmatter}
\title{An optimization approach to adaptive multi-dimensional capital management}

\author[Korteweg,Rabobank]{G.A.~Delsing\corref{cor}}
\cortext[cor]{Corresponding author. E-mail address: G.A.Delsing@uva.nl}
\author[Korteweg,CWI]{M.R.H.~Mandjes}
\author[Korteweg,Radboud]{P.J.C.~Spreij}
\author[Korteweg,Rabobank]{E.M.M.~Winands}

\address[Korteweg]{Korteweg-de Vries Institute for Mathematics, University of Amsterdam, Science Park 107, 1098 XH Amsterdam, the Netherlands}

\address[Rabobank]{Rabobank, Croeselaan 18, 3521 CB Utrecht, the Netherlands}

\address[CWI]{CWI - National Research Institute of Applied Mathematics and Computer Science, Science Park 123, 1098 XG Amsterdam, the Netherlands}

\address[Radboud]{Radboud University, Heyendaalseweg 135, 6525 AJ Nijmegen, the Netherlands}

\begin{abstract}
Firms should keep capital to offer sufficient protection against the risks they are facing. In the insurance context methods have been developed to determine the minimum capital level required, but less so in the context of firms with multiple business lines including allocation.
The individual capital reserve of each line can be represented by means of classical models, such as the conventional Cram\'{e}r-Lundberg model, but the challenge lies in soundly modelling the correlations between the business lines. We propose a simple yet versatile approach that allows for dependence by introducing a common environmental factor. We present a novel Bayesian approach to calibrate the latent environmental state distribution based on observations concerning the claim processes. The calibration approach is adjusted for an environmental factor that changes over time. The convergence of the calibration procedure towards the true environmental state is deduced.  We then point out how to determine the optimal initial capital of the different business lines under specific constraints on the ruin probability of subsets of business lines. Upon combining the above findings, we have developed an easy-to-implement approach to capital risk management in a multi-dimensional insurance risk model.

\end{abstract}

\begin{keyword}
ruin probability; insurance risk; Bayesian statistics; optimal allocation; multi-dimensional risk process\\
\vspace{\baselineskip}
\textit{JEL classification number}: C690; C220
\end{keyword}

\end{frontmatter}

\section{Introduction}
Firms should keep capital so as to be guaranteed a reasonable degree of protection against the risks they face when conducting their business.
In the insurance industry, procedures to find the minimally needed capital level have received a great deal of attention, reflecting the constraints imposed by insurance regulation. For instance, the European solvency regulation. More specifically,  insurance companies should manage their capital reserve level such that the probability of economic ruin within one year is less than a given threshold. This risk measure, known as {\it Value-at-Risk} (VaR), can thus be considered as the key concept when assessing insurance firms' credit risk vulnerability. The main objective of this paper is to develop a strategy to update the firm's risk reserve, and its allocation across different business lines within the firm.

\vb

The capital surplus required to keep the credit risk of a firm sufficiently low, studied in a branch of research known as {\it ruin theory}, depends on various characteristics including 
the distribution of the claim amounts, their inter-arrival times, and the incoming premiums. The focus of ruin theory is on the time evolution of the capital surplus, with its inherent fluctuations due to amounts claimed and premiums earned. We remark that the capital surplus is also a measure of the risk pertaining to a portfolio, and as a consequence the VaR is a relevant concept in the portfolio management context too. 

A traditional objective of risk theory concerns the determination of the initial capital reserve, say $u$, that guarantees the insurer a sufficient level of solvency. Initially, the focus was on the probability $\phi(u)$ of ultimate ruin, i.e.\ the probability that the capital surplus ever drops below zero given the initial reserve $u$; see the seminal contribution~\cite{lundberg1903approximerad}. Later these results have been extended in many ways, most notably (i)~ruin in finite time, (ii)~more advanced claim arrival processes, (iii)~asymptotics of $\phi(u)$ for $u$ large, and (iv) more realistic premium processes (e.g.\ non-deterministic ones); see e.g.\ \cite{MR2766220} for a detailed account.  

\vb

While most of the existing literature primarily considers a univariate setting (focusing on a single reserve process), in practice firms often have multiple lines of 
business. As a consequence, it is a relevant question how to assign initial reserves to the individual business lines, with the objective to keep the firm's credit risk (now expressed in terms of  the likelihood of the capital surplus of one or more of the business lines dropping below zero) sufficiently low. 
A multi-dimensional risk model is introduced by assigning a risk process to each business line.
The allocation of the initial reserve of the firm to its business lines follows directly from the individual initial reserves in this multi-dimensional risk model.
A complication however is that the individual capital surplus processes are typically highly correlated, as they are affected by common environmental factors (think of the impact of the weather on health insurance and agriculture insurance). 

This paper has several contributions. In the first place we set up a simple yet versatile
multivariate risk model, in which the components are correlated by using a common (but unobserved) environmental factor.  In the second place, we develop a Bayesian technique which facilitates the calibration of the environmental factor by observing the claim processes. For a changing environmental factor, we propose a maximum likelihood calibration method. In the third place, we point out how the above ingredients can be used to set up a procedure for periodically adapting the capital reserves based on new observations from the claim processes.

\vb

We proceed with a few more words on the related literature, and its relation to our work. 
Multivariate risk processes play a prominent role in various studies (see e.g.\ the overview \cite[Ch. XIII9]{MR2766220}), but capturing the corresponding joint ruin probability has proven challenging (see e.g.\   \cite{MR2322127,MR2016771}).
Our work is inspired by earlier work by Loisel {\it et al.} \cite{Loisel2004_unpublished,Loisel2007a_unpublished,Loisel2007b}, which also make use of an environmental factor. The main difference is that Loisel assumes a Markov  environmental state factor, whereas in our Bayesian setup the objective is to track the unobservable environmental state. As a fixed environmental state is not realistic over longer time intervals, we point out how to adapt the calibration procedure to detect a change in the environment. Knowing the environmental state, we can compute (or approximate) the ruin probabilities for any given initial capital reserve, which enables the selection of appropriate initial levels. Our procedure also includes a provably converging  Bayesian calibration; recall that the environmental state cannot be observed. In this respect we note that we found only few contributions on this topic that also cover the calibration; an example is \cite{MR2949452}, but the Bayesian updating approach that is proposed there focuses on a single insurer only.
When the environmental state factor is re-sampled each time period, the calibration method has to be adjusted to a maximum likelihood approach in order to achieve convergence towards the distribution.

Most actuarial and financial literature on the topic of capital allocation; for example \cite{LG2004}, focuses on the subdivision of an exogenously given amount of capital for the entire firm over its business lines. Our work differs from this traditional problem of capital allocation within a firm by minimizing the sum of the initial reserves of its business lines. The initial reserve of the firm as well as the initial reserves of its business lines therefore follow directly from this multi-dimensional model and no additional capital allocation procedure is required.
Again, in line with our earlier remark, in this paper the focus is on an insurance context, 
but the framework developed has various other evident applications.
A similar procedure may, for example, be adopted in banking. Banks have some fixed income streams such as interest rate payments on mortgages and loans and the outgoing claims may represent counterparty defaults. In this setting the ruin model can be used to assess credit risk for the portfolio of a bank.

\vb

This paper is organized as follows. Section \ref{sec_model} presents the model and preliminaries. It defines the risk process for each individual business line, and characterizes the finite time probability of ruin in case the environmental factor (and thus also the claim inter-arrival and claim-size distribution) is known. We then present the multivariate insurance model by introducing the environmental dependence.
The section concludes by developing a procedure to allocate capital to the individual business lines under a constraint on the VaR, which is achieved by periodically adapting the capital reserves. 
Section \ref{sec_calibration} introduces a calibration approach for the multivariate risk process of Section \ref{sec_model}, which is geared towards learning the environmental factor based on the claim processes. A Bayesian updating approach is presented for the environmental state factor which does not change (drastically) over time. For an environmental state factor that is re-sampled each observation period from a discrete distribution, we propose a maximum likelihood approach to calibrate the distribution. 
Numerical examples of the capital allocation and calibration approach of the multi-dimensional risk process are given in Section \ref{sec_examples}, including the use of Arfwedson's approximation of the probability of ruin in case there is no explicit solution available. 
Section \ref{sec_conclusion} concludes this paper, and discusses possible extensions of the model.

\section{A multivariate risk model}\label{sec_model}
As pointed out in the introduction, our main objective is to set up a procedure that guarantees a business to stay solvent with a certain degree of confidence over a time horizon $T$ (say). This we achieve by periodically adapting the risk reserves of the business lines. To manage the process, we therefore need a procedure to compute the probability that, given a certain initial reserve level, one or more of the reserve processes drops below $0$ before a specified time $T$. We assume no impact of insolvency of one business line on the others.
Each line of business is free of expenses, taxes and commissions. 
For each of the business lines, there is some initial capital reserve, increase due to premiums (that come in at a fixed rate per unit time), and decrease due to claims. 

We use a multi-dimensional variant of the classical Cram\'{e}r-Lundberg model with $n\in{\mathbb N}$ business lines. Let us now define the dynamics of the capital surplus $X_i(\cdot)$ of business line $i$. There is a constant premium rate $r_i\geq 0$ per unit time. The number of claims arriving in $[0,t]$, denoted by $N_i(t)$, is a Poisson process with parameter $\lambda_i$. The claim sizes $C^i_k$ form a sequences of i.i.d. random variables distributed as random variable $C^i$, with moment generating function $\hat{B}_i[s]$ and distribution function $F_i$. It means that the capital surplus process $X_i(t)$ for business line $i$ is given by
 \begin{equation}\label{eq_ruin}
X_i(t):=u_i+r_i t - \sum_{k=1}^{N_i(t)}C^i_k,
\end{equation}
where $u_i\geq 0$ denotes the initial capital reserve. 
The probability of ruin of business line $i$ before time $T$ is given by
\[\phi_i(u_i,T):=\mathbb{P}\left(\inf_{t\in[0,T]}X_i(t)<0\,\big|\,X_i(0)=u_i\right).\]
In Section \ref{subsec_fixedmodel} we assume that $r_i$, $\lambda_i$ and $F_i$ are given; later, in Section \ref{subsec_dependence}, we introduce a mechanism in which they are randomly selected (in a specific coordinated manner), thus rendering the processes $X_i(\cdot)$ dependent.

\subsection{Model under fixed parameter setting and no dependence}\label{subsec_fixedmodel}
In this section, $r_i$, $\lambda_i$ and $F_i$ are given. In addition, for now the business lines are assumed independent. 
Following classical ruin theory we denote \[\kappa_{i}(s):=\lambda_{i}\left(\hat{B}_i[s]-1\right)-r_i s.\] 
This function is strictly convex (easily deduced by the definition of a moment generating function). Under the {\it net profit condition} $\kappa_{i}'(0)=\lambda_{i}\mathbb{E}[C^i]-r_i<0$ (and a mild regularity assumption: $\kappa_{i}(s)$ should not jump from a value below 0 to $\infty$), it can be shown that a unique positive root $\gamma_i$ of $\kappa_{i}(s)=0$ exists. This root plays a crucial role in Arfwedson's approximation of $\phi_i(u_i,T)$ \cite{MR0074725}; see Appendix \ref{appA}.
 For some specific claim size distributions, the probability of ruin $\phi_i(u_i,T)$ can be explicitly calculated. The proposition below concerns the case of exponentially distributed claims.
 
\begin{proposition}\label{prop_exp} Assume $C^{i}\sim {\rm \exp}(\theta_i)$. Then,
$$\phi_i(u_i,T)=\begin{cases}
 {\displaystyle  \frac{\lambda_{i}}{\theta_i r_i}} \exp\Bigg\{-\left(\theta_i-{\displaystyle \frac{\lambda_{i}}{r_i}}\right)u_i\Bigg\}-{\displaystyle \frac{1}{\pi}}{\displaystyle \int_0^\pi} {\displaystyle \frac{f_1(\mu)f_2(\mu)}{f_3(\mu)}}{\rm d}\mu, & \text{for } \theta_i r_i>\lambda_{i}\\
  &\vspace{-3mm}\\
  1-{\displaystyle \frac{1}{\pi}}{\displaystyle \int_0^\pi} {\displaystyle \frac{f_1(\mu)f_2(\mu)}{f_3(\mu)}}{\rm d}\mu, &\text{for } \theta_i r_i\leq\lambda_{i}
\end{cases}$$
where
\begin{align*}
f_1(\mu)&=\frac{\lambda_{i}}{\theta_i r_i}\exp\Bigg\{2 T\sqrt{\theta_i r_i\lambda_{i}}\cos\mu-\left(r_i\theta_i+\lambda_{i}\right)T+u_i\theta_i\left(\frac{\sqrt{\lambda_{i}}}{\sqrt{r_i\theta_i}}\cos\mu-1\right)\Bigg\},\\
f_2(\mu)&=\cos\left(\frac{u_i\sqrt{\theta_i\lambda_{i}}}{\sqrt{r_i}}\sin\mu\right)-\cos\left(\frac{u_i\sqrt{\theta_i\lambda_{i}}}{\sqrt{r_i}}\sin\mu+2\mu\right),\\
f_3(\mu)&=1+\frac{\lambda_{i}}{\theta_i r_i}-2\frac{\sqrt{\lambda_{i}}}{\sqrt{\theta_i r_i}}\cos\mu.
\end{align*}
\end{proposition}
\begin{proof}
The proof follows from Barndorff-Nielsen and Schmidli \cite{MR1366823} and the observation that the case $\theta_i \neq 1$ can be deduced from the case $\theta_i=1$ via 
\[\phi_{i,\lambda_{i},\theta_i}(u_i,T)=\phi_{i,\lambda_{i}/\theta_i,1}(\theta_i u_i,\theta_iT).\]
The case  $r_i\neq 1$ follows from \[\phi_{i,\lambda_{i},r_i}(u_i,T)=\phi_{i,\lambda_{i}/r_i,1}(u_i,r_iT).\]
This proves the claim.
\end{proof}

Denote $S_1$ up to $S_M$ as specific subsets of the $n$ business lines, for $m\in \mathbb{N}$. We focus on the probability of ruin of all business lines within subset $S_m$.
As the business lines are (for now) assumed independent, 
\begin{equation}\label{eq_pi}
\pi_{m}(u,T):=\mathbb{P}\left(\sup_{i\in S_m}\inf_{t\in[0,T]} X_i(t)<0\Big| X(0)=u\right)=\prod_{i\in S_m}\phi_i(u_i,T).
\end{equation}
Likewise, we could consider the probability of at least one defaulting business line within a subset:
\begin{equation}\label{eq_pi2}
\bar{\pi}_{m}(u,T):=1-\mathbb{P}\left(\inf_{i\in S_m}\inf_{t\in[0,T]} X_i(t)>0\Big| X(0)=u\right)=1-\prod_{i\in S_m}(1-\phi_i(u_i,T)).
\end{equation}
Even though we assumed independence between the different business lines, there can be dependence across the subsets $S_m$ when a business line is contained in multiple sets $S_m$.

\subsection{Environmental Dependence}
\label{subsec_dependence}
We now point out how we can make the processes $X_i(\cdot)$ dependent by working with a common  environmental factor affecting all business lines (think, for example, of the weather impacting the claim process of health-related business lines, but also of business lines related to the agricultural sector). 
{\it Conditional on}  the state of the environment, the multivariate claim process is modelled as the $n$-dimensional process $X_1(\cdot),\ldots,X_n(\cdot)$ defined in the previous subsection; in particular, they are conditionally independent. 

In more concrete terms, our process is defined as follows.  
The environment state, denoted by $P$, is a random variable with support $A=\{1,...,J\}$ (and corresponding probabilities $p_j, \ j\in A$). 
If $P=j$, then the claim arrival rate of business line $i$ is $\lambda_{ij}$, and the claims of business line $i$ are distributed as a random variable $C^{ij}$ (and distribution function $F_{ij}$). 
Conditional on the environmental state, the $X_i(\cdot)$ are independent, so that Equation \eqref{eq_pi} becomes
\[\pi_{m}(u,T)=\sum_{j=1}^J p_j\prod_{i\in S_m}\phi^j_i(u_i,T),\]
with
\[\phi_i^j(u_i,T):=\mathbb{P}\left(\inf_{t\in[0,T]}X_i(t)<0\,\Big|\,X_i(0)=u_i,P=j\right).\]
Here the $\phi_i^j(u_i,T)$ are as the $\phi_i(u_i,T)$ that we defined before, but now with the $\lambda_{ij}$ and $F_{ij}$ being used.

\subsection{Optimal Capital Reserve Allocation}\label{subsec_allocation}
In this subsection we further detail our objective: finding appropriate values of the initial reserves $u_1,\ldots,u_n$, such 
that a VaR-type risk measure remains below some maximally allowed value. 

For a univariate risk process, say that of business line $i$, the conventional setting is that the minimal initial reserve $u_i$ is determined such that the probability of ruin over a specified time horizon remains below a given $\delta\in(0,1)$. We now extend this to the multivariate risk setting introduced above, by considering the ruin probabilities of the specific subsets $S_1,...,S_M$ . 
For $\delta_m\in(0,1)$ (with $m=1,\ldots,M$)  we focus on the optimization problem
\begin{equation}\label{risk_measure}
\min_{u\succeq 0}\sum_{i=1}^n u_i, \ \ \ {\rm subject\:\: to} \ \ \ \pi_{m}(u,T)\leq \delta_m, \:\:\ m=1,...,M,
\end{equation}
where $u\succeq 0$ means that the vector $u$ is component-wise positive.
Evidently other objectives can be chosen, such as constraints on the probability $\bar{\pi}_m$ or
\begin{equation}\label{risk_measure2}
\min_{u\succeq 0}\sum_{i=1}^n u_i,  \ \ \ {\rm subject\:\: to} \ \ \ \mathbb{P}\left(\inf_{t\in[0,T]}\sum_{i=1}^n X_i(t)<0\Big| X(0)=u\right)\leq \delta;
\end{equation}
these can be dealt with in a similar way. 

The environmental state is not observed, so that that the calibration is not straightforward. 
We develop an easy-to-implement Bayesian updating procedure that is  based on the observed claim processes (corresponding to the various business lines). For an environmental state factor that is re-sampled each time period, we propose a maximum likelihood calibration approach. 
Evidently, these procedures should be such that the estimates of the state probabilities $p_j$ can be  updated on a regular basis. The next section presents our approaches.

\section{Detection of the environmental state}\label{sec_calibration}
This section provides an easy-to-implement calibration approaches to path-wise track the unobservable environmental state based on observed claims. It is assumed that the claim intensities and claim size distributions are known. Let $t_0=0<t_1<...<t_M$ and denote $\bar{t}_m:=t_m-t_{m-1}$. The number of claims $Y_i^m:=N_i(t_m)-N_i(t_{m-1})$ and the sequence of claim sizes $Z_i^m=(C_{1}^i,...,C_{Y_i^m}^i)$ during time interval $(t_{m-1},t_{m}]$ are observed for each business line $i$. 
We introduce the notation $\mathcal{Y}^m:=\{Y^1,...,Y^m\}$ and $\mathcal{Z}^m:=\{Z^1,...,Z^m\}$, where $Y^m:=(Y_1^m,...,Y_n^m)$ and $Z^m:=(Z_1^m,...,Z_n^m)$ denote the vectors containing the number of claims and claim sizes for all business lines during time interval $(t_{m-1},t_m]$, respectively. With slight abuse of notation we use the generic notation $f$ to denote the (joint) density of any random quantity. For instance, $f(\mathcal{Z}^m,\mathcal{Y}^m)$ denotes the joint density of the number of observed claims and claim sizes up to time $t_m$.

In Section \ref{subsec_stationary} we assume that the environmental state factor $P$ is fixed over time; later, in Section \ref{subsec_non_stationary}, we introduce a calibration approach for an environmental state that is randomly selected each observation period.
\subsection{Bayesian calibration for time-independent environmental state}\label{subsec_stationary}
In this section the environmental state random variable $P$ is considered independent of time and therefore is not subject to change over time. The environmental state probabilities $p_j$ are estimated as the posterior distribution 
based on the observed claims after some time $t_m$ (say):

\begin{align}
\hat{p}_j^{m}:&=\mathbb{P}(P=j|\mathcal{Y}^m,\mathcal{Z}^m)\nonumber \\
&=\frac{\hat{p}^{0}_{j}f(\mathcal{Y}^m,\mathcal{Z}^m|P=j)}{\sum_{k=1}^J \hat{p}^{0}_{k}f(\mathcal{Y}^m,\mathcal{Z}^m|P=k)},\label{eq_Bayes}
\end{align}
where $\hat{p}^{0}_{j}\in(0,1)$ denote the prior probabilities which can be chosen arbitrarily such that $\sum_{j=1}^J\hat{p}^{0}_{j}=1$.
If we furthermore assume that the observations in each time period are independent, the environmental state probabilities can be estimated iteratively using \eqref{eq_Bayes}:

\begin{align}
\hat{p}_j^{m}:&=\frac{\hat{p}^{0}_{j}f(\mathcal{Y}^{m-1},\mathcal{Z}^{m-1}|P=j)f(Y^m,Z^{m}|P=j)}{\sum_{k=1}^J \hat{p}^{0}_{k}f(\mathcal{Y}^{m-1},\mathcal{Z}^{m-1}|P=k)f(Y^m,Z^{m}|P=k)} \nonumber\\
&=\hat{p}_j^{m-1}f(Y^m,Z^{m}|P=j)\frac{\sum_{l=1}^J \hat{p}^{0}_{l}f(\mathcal{Y}^{m-1},\mathcal{Z}^{m-1}|P=l)}{\sum_{k=1}^J \hat{p}^{0}_{k}f(\mathcal{Y}^{m-1},\mathcal{Z}^{m-1}|P=k)f(Y^m,Z^{m}|P=k)}\nonumber \\
&=\hat{p}_j^{m-1}\frac{f(Y^m,Z^{m}|P=j)}{\sum_{k=1}^J \hat{p}^{m-1}_{k}f(Y^m,Z^{m}|P=k)}\label{eq_Bayes_Iter}.
\end{align}

\begin{example}\label{ex_exp}
Consider the instance in which only the arrival intensities of the claim processes are dependent on the state of the environment. Conditional on the environmental state, the arrival intensity is fixed, i.e.\ $\mathbb{P}(\lambda_i=\lambda_{ij}\mid P=j)=1$. Note that, conditional on the environmental state, the claims processes are independent.

Using \eqref{eq_Bayes_Iter} we find in this case:
\begin{align*}
\hat{p}_j^{m}:=\mathbb{P}(P=j|\mathcal{Y}^m)&=\frac{\hat{p}^{0}_{j}\mathbb{P}(\mathcal{Y}^m|P=j)}{\sum_{k=1}^J \hat{p}^{0}_{k}\mathbb{P}(\mathcal{Y}^m|P=k)}\\
&=\hat{p}^{m-1}_j\frac{\mathbb{P}(Y^m= y^m|P=j)}{\mathbb{P}(Y^m=y^m)}\\
&=\frac{\hat{p}^{m-1}_{j}\prod_{i=1}^n e^{-\lambda_{ij}\bar{t}_m}\lambda_{ij}^{y^m_i}}{\sum_{k=1}^J \hat{p}^{m-1}_{k}\prod_{i=1}^n e^{-\lambda_{ik}\bar{t}_m}\lambda_{ik}^{y^m_i}}.
\end{align*}

Next, we include the influence of the environmental state on the claim size distribution assuming exponentially distributed claims with rate $\theta_i$. Conditional on the environmental state the rate is fixed i.e.\ $\mathbb{P}(\theta_i=\theta_{ij}\mid P=j)=1$. This gives:

\begin{align*}
\hat{p}^{m}_j=&\frac{\hat{p}^{m-1}_{j}\prod_{i=1}^n e^{-\lambda_{ij}\bar{t}_m}\lambda_{ij}^{y^m_i}f(Z_i^m=z_i^m|Y_i^m=y_i^m,P=j)}{\sum_{k=1}^J \hat{p}^{m-1}_{k}\prod_{i=1}^n e^{-\lambda_{ik}\bar{t}_m}\lambda_{ik}^{y^m_i}f(Z_i^m=z_i^m|Y_i^m=y_i^m,P=k)}\\
=&\hat{p}^{m-1}_{j}\frac{\prod_{i=1}^n e^{-\lambda_{ij}\bar{t}_m}(\lambda_{ij}\theta_{ij})^{y^m_i}e^{-\theta_{ij}\sum_{l=1}^{y_i^m} z_{il}^m}}{\sum_{k=1}^J \hat{p}^{m-1}_{k}\prod_{i=1}^n e^{-\lambda_{ik}\bar{t}_m}(\lambda_{ik}\theta_{ik})^{y^m_i}e^{-\theta_{ik} \sum_{l=1}^{y^m_i}z_{il}^m}}.
\end{align*}

Note that this procedure only requires the \textit{total} claim size over a time period for each business line, i.e.\ $\sum_{l=1}^{y_i^m}z_{il}^m$.
\qed
\end{example}

The estimated probability distribution of the environmental factor characterised by the probabilities $\hat{p}_j^m$ is denoted by $\hat{P}^m$.
After every time interval $(t_{m-1},t_m]$ the Bayesian procedure described by formula \eqref{eq_Bayes_Iter} allows for an update of the estimated probability distribution of $P$ based on observed claims during the time interval. As a result capital reserves can be recalculated based on this new estimation.

The estimation of the distribution retrieved from the Bayesian updating procedure, $\hat{P}^m$ converges in probability towards the true distribution of $P$ as $m$ goes to infinity. 
This result follows from Ghosal et al.\ \cite{ghosal2000}, Theorem 5.1. 
To retrieve the true environmental state factor, it is important that the model is identifiable, i.e. different parameter values correspond to different distributions of processes $X_i$. Example \ref{ex_Loisel2} in Section \ref{sec_examples} shows what happens in case this condition is not satisfied.

\subsection{Maximum likelihood calibration approach for environmental state dependence under re-sampling}\label{subsec_non_stationary}
The previous subsection provided a calibration approach for an environmental state that is assumed not to be subject to change over time. In practice, environmental influence and dependence can rarely be considered fixed over time. The objective of this subsection is to outline a calibration procedure to estimate the environmental state probabilities $p_j$ from observed claims in case the environmental state factor is re-sampled each observation period at random: during observation period $(t_0,t_1]$ the environmental state is then $P_1\in A$, throughout $(t_1,t_2]$ the environmental state factor is $P_2\in A$, etc. In this instance, the observed claims and claims sizes $Y^m, Z^m$ have a (potentially) different underlying environmental state factor for different observation periods such that the Bayesian calibration approach outlined in the previous section has to be adjusted (formula \eqref{eq_Bayes} and \eqref{eq_Bayes_Iter} have to be adjusted) in order to retrieve the distribution of the environmental state factor.

A maximum likelihood approach is adopted to retrieve an estimate of the distribution probabilities $p_j$. Define the maximum likelihood environmental state over observation period $\bar{t}_m$ as:
\begin{align*}
\hat{J}_m:&=\argmax_{j\in \{1,...,J\}} f( Y^m= y^m,Z^m= z^m|P=j)\\
&=\argmax_{j\in J}\prod_{i=1}^n f( Y_i^m= y_i^m,Z_i^m= z_i^m|P=j).
\end{align*}

The probabilities $p_1,..,p_J$ can be estimated after $t_m$ by
$$\hat{p}^{m}_i:=\frac{1}{m}\sum_{k=1}^{m} \mathbbm{1}_{\hat{J}_k=i}.$$

Consistency of the maximum likelihood estimates, i.e.\ convergence in probability of the estimates $\hat{p}^{m}_i$ towards to the true probabilities $p_i$ when $m\rightarrow\infty$, has been shown to hold under specific conditions. One of these conditions concerns the identification of the model, to make sure that different parameter values necessarily correspond to different distributions. The remaining conditions are more technical conditions on the probabilities $p_i$ and the likelihood function $f( Y^m= y^m,Z^m= z^m|P=j)$ and are generally satisfied in practice.
We refer, e.g., to Section 5.5 in \cite{Vaart} for a detailed technical analysis of the consistency conditions.

The environmental factor distribution can be estimated iteratively:
$$\hat{p}^{m}_i=\frac{m-1}{m}\hat{p}^{m-1}_i+ \frac{1}{m}\mathbbm{1}_{\hat{J}_m=i}, \ \ \ \forall i\in\{1,...,J\}.$$

\vb 

We have implemented the above calibration procedure for various examples and obtained initial capital reserves $u_1,\ldots,u_n$ by solving the optimization problems presented in Section \ref{subsec_allocation}. The next section presents the results.

\section{Numerical Results}\label{sec_examples}
In this section, we do not only elaborate on the applicability of the capital updating procedures derived in Sections \ref{subsec_stationary} and \ref{subsec_non_stationary} but also discuss possible extensions. 

The examples in this section are inspired by the numerical setup used by Loisel \cite{MR2145483}. We present four examples highlighting different features of the capital updating procedure and the impact on the calibration of the environmental state factor and optimized allocated initial capital reserves. The examples of increasing complexity cover:

\begin{enumerate}[(i)]
\item the volatility and range of the calibration (and optimization) method with respect to the simulation setting;
\item different parameter sets;
\item a changing environmental state factor; and
\item different claim size distributions.
\end{enumerate}
 
All computations were done in R using an implementation (nmkb) of the Nelder-Mead algorithm for the optimization procedure. Nelder-Mead uses function values only, is robust and known to work well for non-differentiable functions. Numerical integrals are evaluated using Simpson’s adaptive quadrature method.

\begin{example}\textit{(Volatility of the calibration and optimization method)}\label{ex_Loisel}

This example illustrates the volatility and range of the estimated environmental state factors as well as the allocated initial capital reserves. We do so by simulating the process (in every run with a different simulation seed). Each run results in a path for the estimated environmental state factor over time. Comparing the output of the different runs, we observe that the convergence of the calibration approach is not dependent on the simulation seed. We note that this example is the only example in this paper that comprises of multiple simulation runs.

Similar to the setting introduced by Loisel \cite{MR2145483}, we consider a business model consisting of two lines of business $(n=2)$ and identify three different states of the economy $(J=3)$. We assume the claims to be exponentially distributed $C^i\sim \exp(\theta_i)$, similar to Example \ref{ex_exp}, and take premium rates $r_1=r_2=1$. In this first example of this section the environmental state does not change over time.
Initially, we also assume $\theta_i=1$ independent of the environmental state. 
The influence of the environmental state factor on the second business line, through the intensity of the claims process, is kept constant using $\lambda_{21}=\lambda_{22}=\lambda_{23}=0.6$. The optimal allocation therefore strongly depends on the claim intensity parameter of the first business line. In the first state of the environment, $\lambda_{11}=0.5<\lambda_{21}$, line of business 1 is safer than line 2 by comparison of the claim intensities and therefore should result in greater capital reserves for the second line of business. In the second and third environmental state, business line 1 is more risky with intensities $\lambda_{12}=0.7$ and $\lambda_{13}=0.92$, respectively. We denote this parametrization of the claim intensities by $\lambda=(0.5,0.6,0.7,0.6,0.92,0.6)$.

The environmental state factor is estimated using the Bayesian calibration method presented in Section \ref{subsec_stationary}. 
The mean and sampled confidence interval of the error of $\hat{p}^m$ with respect to the real environmental factor $P=1$ (i.e.\ $(p_1,p_2,p_3)=(1,0,0)$) are presented in Figure \ref{fig_Errors_sim_Loisel}(a) for $\bar{t}_m=1$ and prior distribution $\hat{p}^0=\left(\frac{1}{3},\frac{1}{3},\frac{1}{3}\right)$.

\vb

\begin{figure}[H]
\centering
      \includegraphics[width=1.0\textwidth]{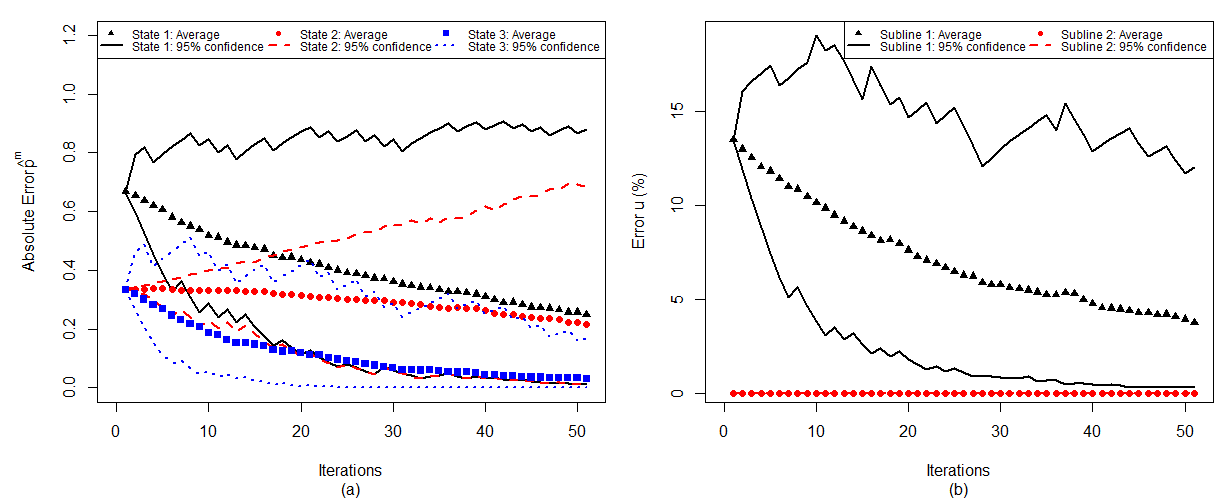}
  \caption{Mean and 95\% confidence interval of absolute error of estimate environmental state factor probabilities $\hat{p}^m=(\hat{p}_1^m,\hat{p}_2^m,\hat{p}_3^m)$ with respect to the true environmental state $P=1$ and relative error of allocated $\hat{u}$ with $\mu=(1,1,1,1,1,1)$ and $\lambda=(0.5,0.6,0.7,0.6,0.92,0.6)$ for 100 trials.}
\label{fig_Errors_sim_Loisel}
\end{figure}

Reserves are allocated by solving minimization problem (\ref{risk_measure}) using $\delta_m=0.001^{|S_m|}$ and fixed $T=1$. We chose $\delta_m=0.001$, which is in the same order of magnitude as the insurance and banking capital regulation thresholds, which use a 0.5\% and 0.1\% confidence level, respectively. In Figure \ref{fig_Errors_sim_Loisel}(b) the mean error using the estimated environmental state probabilities $(\hat{p}_1,\hat{p}_2,\hat{p}_3)$ has been plotted as a fraction of the allocated capital reserves using the true environmental state factor $P=1$. We refer to this as the ``error" of $u$. The figure also shows the 95\% confidence range of the allocated reserves (for 100 optimization runs). As there is no influence of the environmental factor on the reserve process of line 2, we observe no impact on the capital allocation for this business line. 

This example shows the convergence of the Bayesian calibration approach towards the true environmental state and the convergence of the allocated initial capital reserves towards the optimal capital reserves for both business lines. This convergence holds for every random sample.
\end{example}

\begin{example}\textit{(Different parameter sets)}\label{ex_Loisel2}

This example extends the previous example by varying the claim arrival intensity and claim size parameters to determine the impact of these parameters. Increased influence of the environmental state factor on the claims arrival intensity and size, results in faster convergence.
Figures \ref{fig_p_par} and \ref{fig_U_par} show the estimates $\hat{p}^m$ and allocated initial reserves for various sets of intensity and claim size parameters. (The same random seed has been used for the different parameters sets to ensure for a fair comparison.)

With respect to the previous example, the present example also includes dependence of the claim size distribution on the environmental state factor. As the graphs illustrate, this results in faster convergence.
The stronger the dependence of the capital reserve process on the environmental state (through the claim intensity as well as the claim size), the more sensitive the capital allocation. 
When the environmental states have the same impact  on the claim intensity and size, the environmental states are essentially indistinguishable in the model, which can be observed in \ref{fig_p_par}(d) and \ref{fig_U_par}(d). This is a violation of the identification condition for the convergence of the calibration procedure.

\begin{figure}[H]
\centering
      \includegraphics[width=1.0\textwidth,height=400pt]{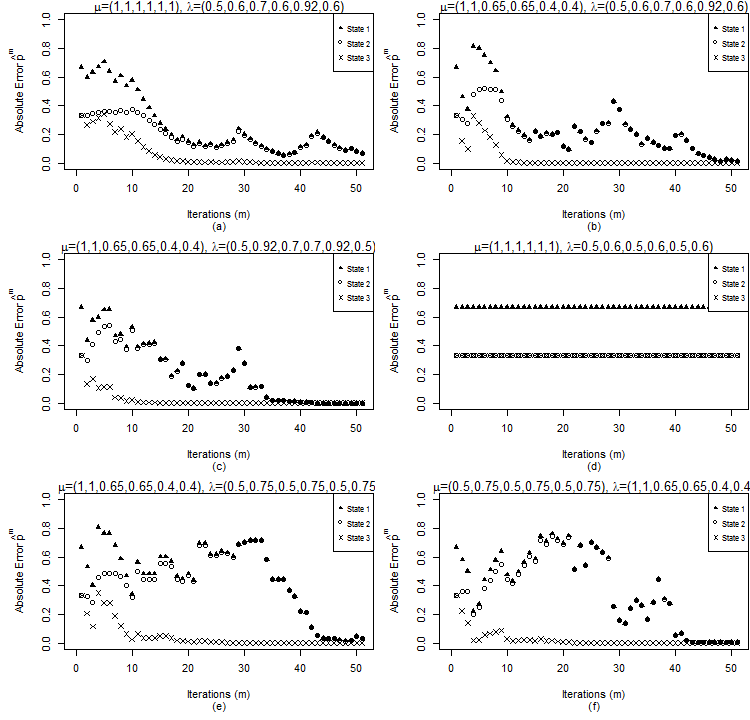}
  \caption{Absolute error of estimate environmental process probabilities $\hat{p}^m=(\hat{p}_1^m,\hat{p}_2^m,\hat{p}_3^m)$ with respect to the true environmental state factor $P=1$ for various parameter sets.}
\label{fig_p_par}
\end{figure}
\newpage
\begin{figure}[H]
\centering
      \includegraphics[width=1.0\textwidth,height=400pt]{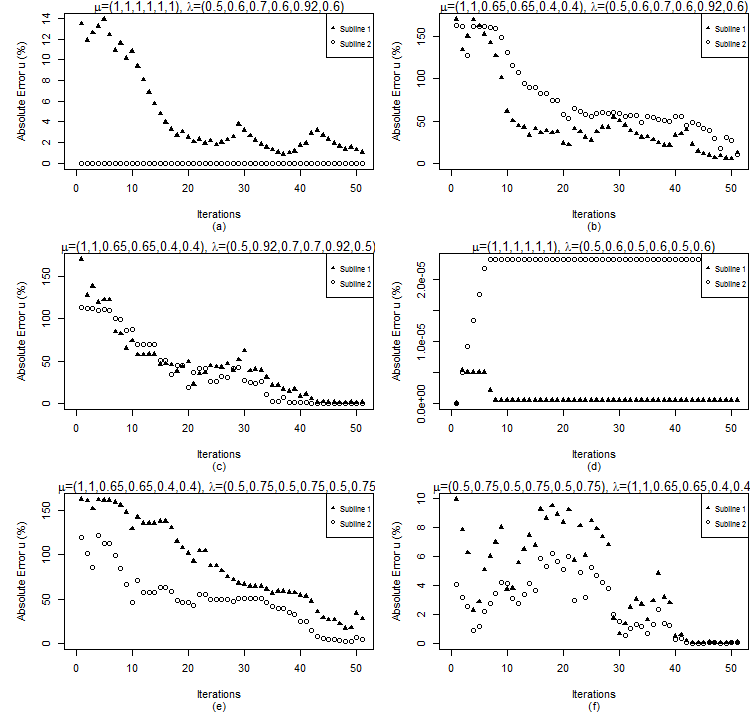}
  \caption{Relative error of allocated capital reserve $\hat{u}$ using the estimated environmental state distribution $\hat{p}^m$ with respect to the allocated capital reserve using the true environmental state factor $P=1$ for various parameter sets.}
\label{fig_U_par}
\end{figure}
\end{example}

\begin{example}\textit{(Changing environmental state factor)}\label{ex_Loisel3}

In practical situations the environmental state factor is not necessarily constant. Therefore we consider in this example an instance where it changes over time. First, we introduce a single change of the environmental state factor by switching the environmental state from 1 to 2 after the 10th time interval. We show that the Bayesian calibration approach still converges to the true environmental state factor over time. Next, we introduce an environmental state factor that changes more frequently over time by re-sampling the environmental state each observation period at random. In this case we apply the calibration method outlined in Section \ref{subsec_non_stationary} to find the true environmental state factor distribution. 
The model setup in this example is the same as in Example \ref{ex_Loisel}. 

\vb

Figure \ref{fig_Switch} shows the results on the estimated $\hat{p}^m$ and allocated initial reserves in case the environmental state factor switches from state 1 to 2 after 10 time intervals of length 1. The Bayesian calibration approach converges towards the new environmental state ($P=2$) over time, see Figure \ref{fig_Switch}(a).

\begin{figure}[H]
\centering
      \includegraphics[width=1.0\textwidth]{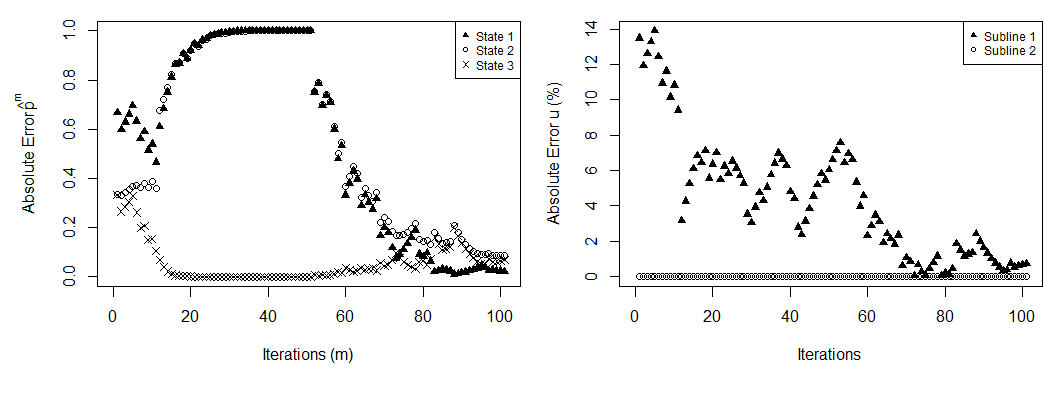}
  \caption{Absolute error of estimate environmental process probabilities $\hat{p}^m=(\hat{p}_1^m,\hat{p}_2^m,\hat{p}_3^m)$ with respect to the true environmental state in case of a switch from $P=1$ to $P=2$ after $t_m=10$ and relative error of allocated $\hat{u}$ with $\mu=(1,1,1,1,1,1)$ and $\lambda=(0.5,0.6,0.7,0.6,0.92,0.6)$.}
\label{fig_Switch}
\end{figure}

Your objective might entail the fast convergence towards the true environmental state or the earlier detection of a changing environmental state. In some cases the updating procedure may have converged towards the environmental state and a change in the environmental factor cannot be detected. Introducing a weighting function $h_w(\cdot):(0,1)\rightarrow\mathbb{R}$ over the previous probability estimates $\hat{p}^{m-1}$ in formula \eqref{eq_Bayes_Iter} may improve the updating procedure. Dependent on your own objective it may increase or decrease the convergence towards the true environmental state factor.
A straightforward example is the power-function: $h_w(\hat{p}_j^{m-1})=(\hat{p}_j^{m-1})^w$ for some fixed constant $w$. In Figure \ref{fig_weight_par} we show the impact of this weighting function on the convergence of the estimated environmental state probabilities $\hat{p}^m$ in case of a switch after 10 time intervals, as before. When choosing $w>1$, high probabilities carry more weight than in the case of no weighting. For a time-independent environmental state this would result in faster convergence towards the true state $P=1$ and subsequently slower adaptation to a potential switch in environment. For $w<1$, the convergence of the environmental state probabilities $\hat{p}^m$ towards the true state $P=1$ is slower than in case of no weighting. However, due to the slower convergence, the new state $P=2$ is recognized faster. Depending on how fast one wants to recognize a new environment state, one might choose a specific weighting function.\\

\begin{figure}[H]
\centering
      \includegraphics[width=1.0\textwidth,height=400pt]{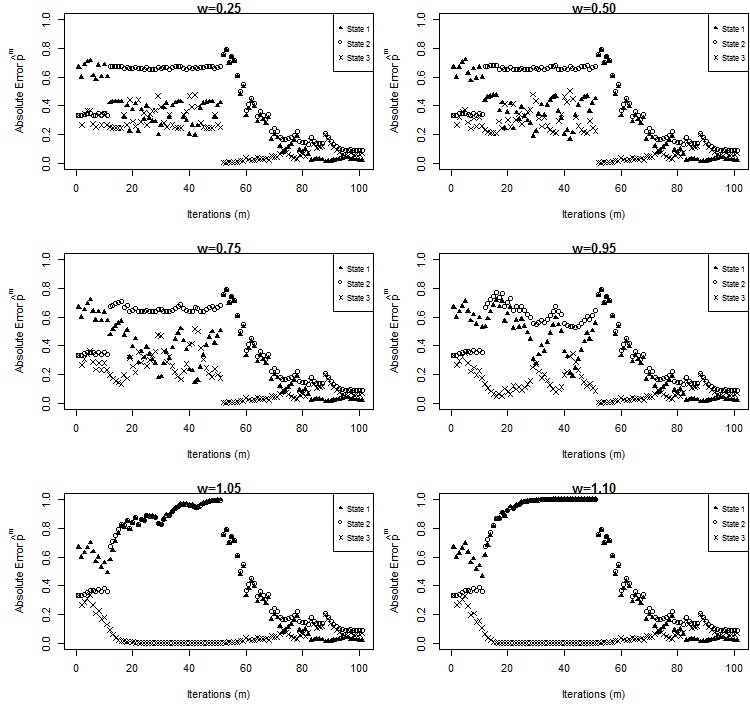}
  \caption{Impact of weighting function $h_w(\hat{p}_j^m)=(\hat{p}_j^m)^w$ on the absolute error of estimate environmental factor probabilities and relative error of allocated capital reserves as a function of $w$ in case of a single switch in environmental state from $P=1$ to $P=2$ at $t_m=10$. Parameters coincide with Figure \ref{fig_Switch}.}
\label{fig_weight_par}
\end{figure}

Next, we re-sample the environmental state factor each observation period (length 1) from the true distribution $p=(1/3,1/3,1/3)$. In this case the Bayesian calibration approach cannot be applied and we make use of the calibration approach outlined in Section \ref{subsec_non_stationary}. Figure \ref{fig_UN_update_max} shows that the calibration converges towards the true environmental state distribution. Furthermore, the initial capital reserves retrieved by solving optimization \eqref{risk_measure} differ very little from the capital reserves allocated using the true environmental state factor distribution. 

The parameters used in the example are given by 
$$\mu=\begin{pmatrix}
  1 & 0.65 & 0.4 \\
  1 & 0.65 & 0.4 
 \end{pmatrix} \ \rm{and} \ \lambda=\begin{pmatrix}
  0.50 & 0.70 & 0.92 \\
  0.92 & 0.70 & 0.50
 \end{pmatrix},$$
 where the $(i,j)$-th element in the matrices corresponds with $\mu_{ij}$ and $\lambda_{ij}$, respectively.

\begin{figure}[H]
\centering
      \includegraphics[width=1.0\textwidth]{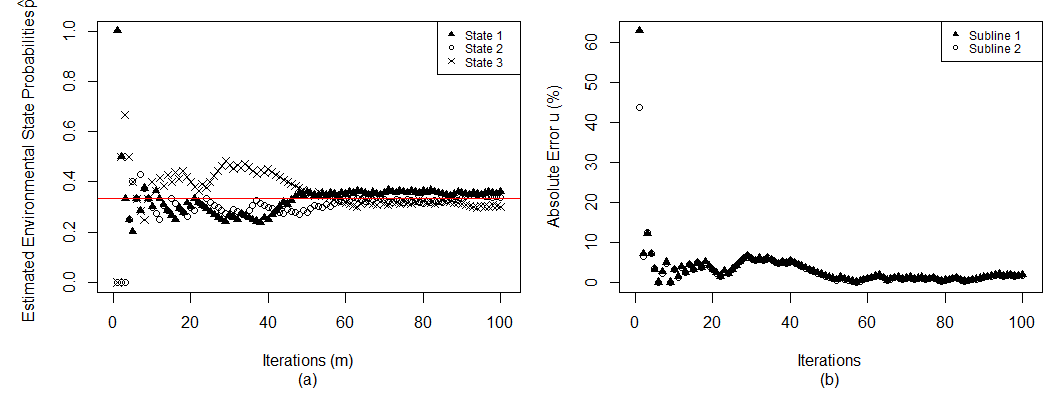}
  \caption{Estimated environmental factor probabilities $\hat{p}^m=(\hat{p}_1^m,\hat{p}_2^m,\hat{p}_3^m)$ over time and relative error of initial capital reserves $\hat{u}$ using the calibration approach in Section \ref{subsec_non_stationary}.}
\label{fig_UN_update_max}
\end{figure}
\end{example}

\begin{example}\textit{(Non-exponential claim size distributions)}\label{ex_Gaussian}

This example relaxes the assumption of exponentially distributed claims by allowing for other claim size distributions, thereby granting the model more flexibility. In the insurance context the capital level has negative jumps (claims). By allowing negative claim sizes, the model could be used for firms that have uncertain incoming cash flows (due to derivative investments for example). In this example we assume a Gaussian distribution for these claim sizes, i.e. $C^i\sim \mathcal{N}(\mu_{i},\sigma_i)$.

No explicit expression exists for the finite time ruin probability for a risk process with Gaussian distributed claims and therefore we use Arfwedson's approximation to estimate these probabilities, see Appendix \ref{appA}. 	
The Bayesian calibration approach for the environmental state distribution outlined in formula \eqref{eq_Bayes_Iter} is then given by:\begin{align*}
\hat{p}^{m}_j=&\frac{\hat{p}^{m-1}_{j}\prod_{i=1}^n e^{-\lambda_{ij}\bar{t}_m}\lambda_{ij}^{y^m_i}f(Z_i^m=z_i^m|Y_i^m=y_i^m,P=j)}{\sum_{k=1}^J \hat{p}^{m-1}_{k}\prod_{i=1}^n e^{-\lambda_{ik}\bar{t}_m}\lambda_{ik}^{y^m_i}f(Z_i^m=z_i^m|Y_i^m=y_i^m,P=k)}\\
=&\hat{p}^{m-1}_{j}\frac{\prod_{i=1}^n e^{-\lambda_{ij}\bar{t}_m}\lambda_{ij}^{y^m_i}\frac{1}{\sigma_{ij}}e^{-\frac{1}{2\sigma_{ij}^2}\sum_{l=1}^{y_i^m} (z_{il}^m-\mu_{ij})^2}}{\sum_{k=1}^J \hat{p}^{m-1}_{k}\prod_{i=1}^n e^{-\lambda_{ik}\bar{t}_m}\lambda_{ik}^{y^m_i}\frac{1}{\sigma_{ik}}e^{-\frac{1}{2\sigma_{ik}^2}\sum_{l=1}^{y_i^m} (z_{il}^m-\mu_{ik})^2}}.
\end{align*}

Our aim is to allocate capital reserves over 5 different business lines within a firm when there are 5 different states of the environment. 
Consider $r_i=1$, $\mu_{i}=1$, $\sigma_i=1$, for all business lines $i$ and introduce environmental state dependence on the claim intensity by setting $$\lambda=\begin{pmatrix}
  0.709 & 0.544 & 0.609 & 0.536 & 0.580 \\
  0.611 & 0.537 & 0.588 & 0.541 & 0.725 \\
  0.730 & 0.601 & 0.636 & 0.620 & 0.691  \\
  0.639 & 0.605 & 0.638 & 0.713 & 0.591 \\
  0.637 & 0.615 & 0.600 & 0.623 & 0.740 
 \end{pmatrix}.$$
 
Figure \ref{fig_Normal_p} shows the results for the estimates $\hat{p}^m$. The true environmental state used in this example is $P=1$ with observation periods of length $1$ and prior distribution $\hat{p}^0=\left(\frac{1}{5},\frac{1}{5},\frac{1}{5},\frac{1}{5},\frac{1}{5}\right)$. The figure shows the fast convergence towards the true environmental state. In general, we observe faster convergence of the Bayesian calibration approach when there are more business lines due to the fact that we then have more observations each observation period (one for each business line).

\begin{figure}[H]
\centering
      \includegraphics[scale=0.7]{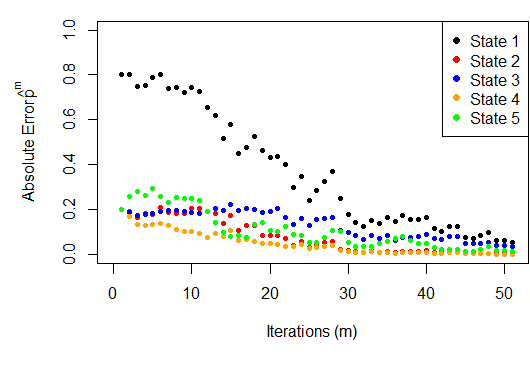}
  \caption{Convergence of absolute error of estimate environmental factor probabilities $\hat{p}^m=(\hat{p}_1^m,\hat{p}_2^m,\hat{p}_3^m,\hat{p}_4^m,\hat{p}_5^m)$ with respect to the true environmental state $P=1$.}
\label{fig_Normal_p}
\end{figure}

Capital reserves are allocated by solving two different minimization problems: minimization problem \eqref{risk_measure} with constraints on the ruin probability of all business lines in a subset, $\pi_m$ as in Equation \eqref{eq_pi}, and problem \eqref{risk_measure} with constraints on the probability of ruin of at least one business line in the subset, $\bar{\pi}_m$ as in Equation \eqref{eq_pi2}. The results are depicted in Figure \ref{fig_Normal_u}(a) and \ref{fig_Normal_u}(b), respectively. (In line with the previous example we have used constraints $\delta_m=0.001^{|S_m|}$.) Constraining the probability of ruin of all business lines in a subset, the optimal allocated initial capital reserves $u$ (in case the environmental state is known) are given by $(20.620,14.553,22.507,15.985,15.869)$. Putting a constraint on the probability of ruin of at least one business line in a subset leads to optimal allocated initial capital reserves $u$ of $(116.173, 150.040, 119.690, 83.281, 108.053)$.

\begin{figure}[H]
\centering
      \includegraphics[width=1.0\textwidth]{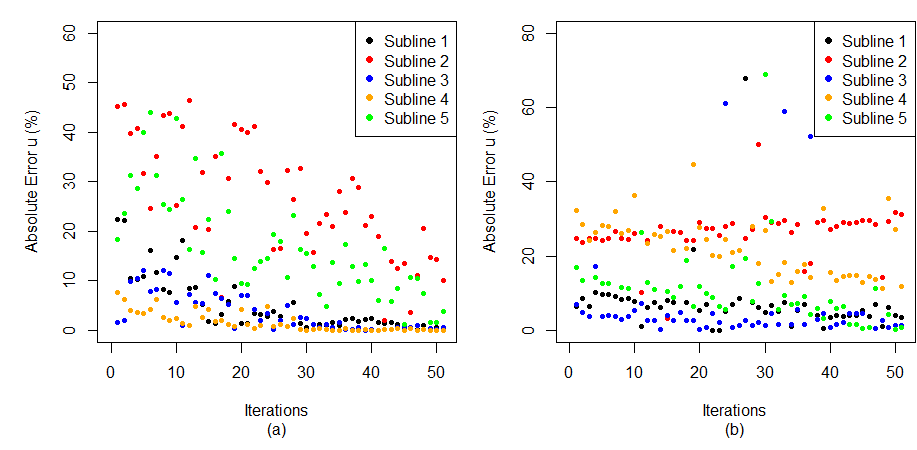}
  \caption{Convergence of the relative error of allocated $\hat{u}$ using the estimated environmental state distribution $\hat{p}^m$ with respect to the allocated capital reserve using the true environmental state factor $P=1$.}
\label{fig_Normal_u}
\end{figure}

\end{example}

\begin{remark}
During our numerical study we have made some general observations concerning the calibration of the environmental state factor and the optimization of the initial capital reserves. These observations include the faster convergence when defining more business lines. This property also applied when there is a more pronounced impact of the environmental state  on the claims intensity and claim size. As shown in Example \ref{ex_exp}, assuming exponentially distributed claims the Bayesian calibration procedure only requires the total number of claims and claim sizes (sum of all claim sizes) per business line for each observation period. We do not need the exact size or timing of each individual claim. 

Furthermore, Figure \ref{fig_Convexity} shows that the ruin probabilities of the business lines under each environmental state, calculated using Proposition \ref{prop_exp}, are convex in $u$.
By definition, we then have a convex optimization problem (\ref{risk_measure}) and a global minimum in $u$ must satisfy the Karush-Kuhn-Tucker (KKT) conditions.

\begin{figure}[H]
\centering
      \includegraphics[scale=0.4]{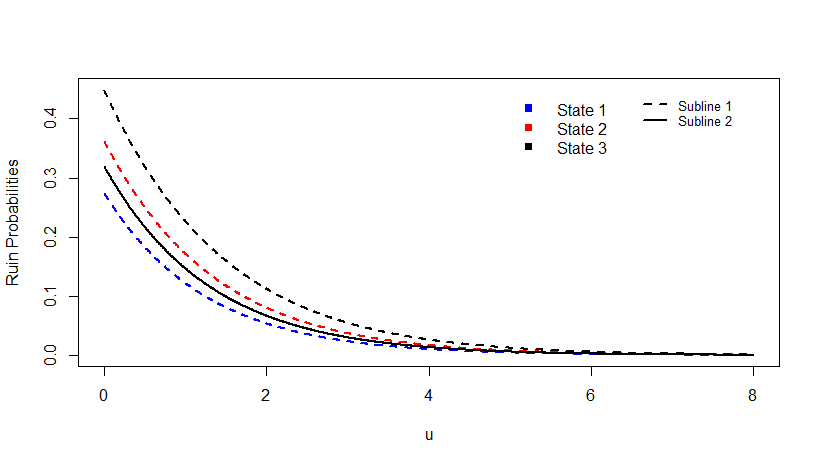}
  \caption{Ruin probabilities as a function of the initial reserve $u$ for $\mu=(1,1,1,1,1,1)$ and $\lambda=(0.5,0.6,0.7,0.6,0.92,0.6)$.}
\label{fig_Convexity}
\end{figure}
\end{remark}

\section{Conclusion and outlook}\label{sec_conclusion}

A multi-dimensional insurance risk model has been introduced for the purpose of allocating capital reserves across different lines of business within a firm. The individual risk process of each business line is given by the Cram\'{e}r-Lundberg model. To model dependence between different business lines, we have introduced a common environmental factor.
Due to the unobservable nature of this factor, we have presented a novel Bayesian approach to calibrate the latent environmental state distribution based on the claim processes and adapted the approach for an environmental state factor that is re-sampled each observation period. The convergence of these calibration approaches towards the true environmental state distribution has been deduced from known results.
Appropriate initial capital reserves are found by solving a constraint optimization problem. Allocation of the capital reserves over the business lines follows as a result from the optimization itself. Numerical examples illustrating the capital allocation technique and Bayesian calibration of the environmental state factor have been presented. We did not only elaborate on the applicability of the derived capital updating procedure but also discussed possible ways of extending the procedure. This includes the possible use of a weighting function to improve the updating procedure.

\vb
We have considered an environmental factor changing over time by re-sampling the factor each observation period. While it is difficult to predict when a change in environment might occur, the environmental state factor is unlikely to be re-sampled (independently) each observation period.
This would argue in favour of a Markov environmental factor in which the time spent in an environmental state is exponentially distributed. Under this assumption, the current setup becomes increasingly more complicated and one would most likely have to resort to numerical approaches to sample the multivariate risk process, similar to the works performed by Loisel et al.\ \cite{Loisel2004_unpublished,Loisel2007a_unpublished,Loisel2007b}.
This area of interest is marked for future research.

\section{References}
\bibliography{References}

\newpage
\begin{appendices}
\section{ }\label{appA}
\begin{proposition}{\bf (Arfwedson Approximation)}\\\label{cor_Arfwedson}
For $u_i>0$ define $\alpha_{i}$ and $\beta_{i}$ as the solution to:
$$\kappa_{i}'(\alpha_{i})=\frac{u_i}{T}, \ \ \beta_{i}=\alpha_{i}-\frac{T}{u_i}\kappa_{i}(\alpha_{i})$$
and let $\tilde{\alpha}_{i}<\alpha_{i}$ denote the solution of $\kappa_{i}(\tilde{\alpha})=\kappa_{i}(\alpha_{i})$. Then,
\begin{enumerate}
\item If $r_i>\lambda_{i}\mathbb{E}[C^i]$, then
$$\phi_i(u_i,T)\sim
\begin{cases}
  \tilde{K}_{i}e^{-\beta_{i}u_i}, & \text{for } T<u_i/\kappa '_{i}(\gamma_{i}) \\
     \frac{K_{i}}{2}e^{-\gamma_{i}u_i}, & \text{for } T=u_i/\kappa '_{i}(\gamma_{i}) \\
    K_{i}e^{-\gamma_{i}u_i}+\tilde{K}_{i}e^{-\beta_{i}u_i}, & \text{for }T>u_i/\kappa '_{i}(\gamma_{i})
  \end{cases}, \ \ \ \ u_i\rightarrow\infty$$
\item If $r_i<\lambda_{i}\mathbb{E}[C^i]$, then
$$\phi_i(u_i,T)\sim
\begin{cases}
  \tilde{K}_{i}e^{-\beta_{i}u_i}, & \text{for } T<u_i/\kappa '_{i}(0) \\
     \frac{\alpha_{i}}{2\tilde{\alpha}_{i}}, & \text{for } T=u_i/\kappa '_{i}(0) \\
    \frac{\alpha_{i}}{\tilde{\alpha}_{i}}+\tilde{K}_{i}e^{-\beta_{i}u_i}, & \text{for }T>u_i/\kappa '_{i}(0)
  \end{cases}, \ \ \ \ u_i\rightarrow\infty$$
 \item If $r_i=\lambda_{i}\mathbb{E}[C^i]$, then
 $$\phi_i(u_i,T)\sim \tilde{K}_{i}e^{-\beta_{i}u_i},$$
where
$$K_{i}:=\frac{r_i-\lambda_{i}\mathbb{E}[C^i]}{\lambda_{i}\hat{B'_{i}}[\gamma_{i}]-r_i}, \ \ \ \tilde{K}_{i}:= -\frac{\alpha_{i}-\tilde{\alpha}_{i}}{\alpha_{i}\tilde{\alpha}_{i}\sqrt{2\pi T \lambda_{i}\hat{B''_{i}}[\alpha_{i}]}}.$$

\end{enumerate}
\end{proposition}
\begin{proof}
The proof follows immediately from Arfwedson's paper \cite{MR0074725} (Scheme I on page 78) and the Cram\'{e}r-Lundberg expression for infinite time ruin probabilities.
\end{proof}
In Figure \ref{fig_Arfwedson} we present the performance of the Arfwedson approximation for exponential claims with respect to the numerically evaluated integral expression presented in Proposition \ref{prop_exp} for various parameter sets.
It can be observed that the Arfwedson approximation improves in accuracy whenever the initial capital $u$ or the time horizon tends to be large, as expected.
\begin{figure}[H]
\centering
      \includegraphics[width=1.0\textwidth]{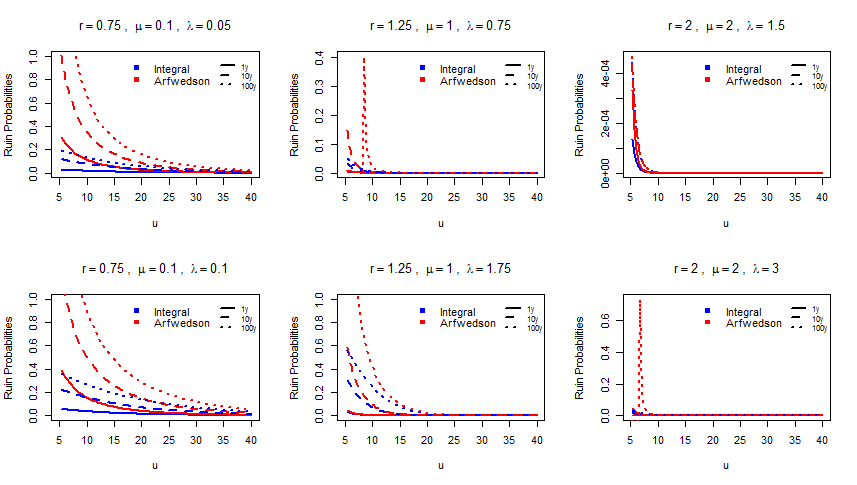}
  \caption{Ruin Probabilities calculated using Arfwedson's approximation and the numerical integral of Proposition \ref{prop_exp} for various parameter sets}
  \label{fig_Arfwedson}
\end{figure}

\end{appendices}

\end{document}